\documentclass[pdf-latex]{sn-jnl}

\usepackage[utf8]{inputenc}
\usepackage{amsmath, amssymb, amsthm, algorithm, algpseudocode, physics}

\usepackage{cite}%
\usepackage{graphicx}
\graphicspath{{./images/}}  
\usepackage{enumerate}

\usepackage{url}
\usepackage[dvipsnames]{xcolor}

\usepackage{lineno}

\DeclareMathOperator*{\argmax}{arg\,max}
\DeclareMathOperator*{\argmin}{arg\,min}
\DeclareMathOperator{\sgn}{sgn}

\newcommand{\set}[1]{\mathcal{#1}}

\newtheorem{theorem}{\bf Theorem} 

\begin{document}
\title{Non-binary dynamical Ising machines for combinatorial optimization}

\author*{\fnm{Aditya} \sur{Shukla}}\email{aditshuk@umich.edu}
\author*{\fnm{Mikhail} \sur{Erementchouk}}\email{merement@gmail.com}
\author{\fnm{Pinaki} \sur{Mazumder}}\email{pinakimazum@gmail.com}

\affil{\orgdiv{Electrical Engineering and Computer Science Department}, \orgname{University of Michigan}, \orgaddress{\city{Ann Arbor}, \postcode{48109}, \state{MI}, \country{U.S.A}}}

\abstract{Dynamical Ising machines achieve accelerated solving of complex combinatorial optimization problems by remapping the convergence to the ground state of the classical spin networks to the evolution of specially constructed continuous dynamical systems. The main adapted principle of constructing such systems is based on requiring that, on the one hand, the system converges to a binary state and, on the other hand, the system's energy in such states mimics the classical Ising Hamiltonian.
%
The emergence of binary-like states is regarded to be an indispensable feature of dynamical Ising machines 
as it establishes the relation between the machine's continuous terminal state and the inherently discrete solution of a combinatorial optimization problem. This is emphasized by problems where the unknown quantities are represented by spin complexes, for example, the graph coloring problem. In such cases, an imprecise mapping of the continuous states to spin configurations may lead to invalid solutions requiring intensive post-processing. In contrast to such an approach, we show that there exists a class of non-binary dynamical Ising machines without the incongruity between the continuous character of the machine's states and the discreteness of the spin states. 
We demonstrate this feature by applying such a machine
to the problems of finding proper graph coloring, constructing Latin squares, and solving Sudoku puzzles.
%
%
Thus, we demonstrate that the information characterizing discrete states can be unambiguously presented in essentially continuous dynamical systems. This opens new opportunities in the realization of scalable electronic accelerators of combinatorial optimization.
}


\maketitle


\section{Introduction} 
\label{sec:intro}

The ever-growing demand for solving complex computational problems compels researchers to explore alternative models of computation based on unconventional principles. Recently, the approach employing classical spin systems started attracting explosive interest. The computational capabilities of such systems were a focus of investigations for a long time owing to the relation between the spin states delivering the lowest energy of the Ising model ~\cite{kirkpatrickOptimization1983, hopfieldNeurons1984, cernyThermodynamical1985, fuApplication1986}, and a broad class of combinatorial optimization problems~\cite{barahonaComputational1982, Lucas2014}.  The novel perspective that motivates the special interest is the realization of the computational capabilities of classical spin systems in continuous dynamical systems, 
which led to the recognition of a particular class of computing devices, the dynamical Ising machines~\cite{Mohseni2022IsingProblems, Bybee2023EfficientMachines, Si2024Energy-efficientProblems}. While based on different underlying dynamical principles, from degenerate optical parametric oscillators~\cite{yamamotoCoherent2017} to bistable dynamics~\cite{Zhang2022AMachine-BRIM}, the dynamical Ising machines share a common operational principle. They employ a characteristic feature of selected continuous dynamical systems to converge to binary or close-to-binary states~\cite{bohmOrderofmagnitude2021}. Ensuring that the energy of such states reproduces the energy of the classical Ising model and that the energy decreases with the machine evolution establishes the connection between the progression of the dynamical Ising machines and finding the ground state of a classical spin system. 

A strong binary structure may appear imperative for the terminal states of the Ising machine. Indeed, substantial deviations from such a structure introduce an uncertainty in the correspondence between the machine's terminal state and the required binary state.
In some situations, this uncertainty does not constitute a fundamental challenge. When any binary state represents a solution to an optimization problem, a ``suboptimal" mapping of an unstructured terminal state to a binary state only impacts the quality of the solution. From a practical perspective, the quality of the solution is an important metric, but it is only one of the metrics characterizing an approach to solving an optimization problem. If other metrics, such as time complexity, scalability, and parallelizability, are favorable, the approach may still be highly beneficial. However, in some situations, only selected binary states may represent a solution. For example, while solving a coloring (labeling) problem, colors are represented by spin complexes. Not all states of those complexes may correspond to a valid color. From the perspective of the original problem, this means that the ``suboptimal" mapping may produce no solution at all, which raises doubt about whether the approach is a viable option for solving high-level optimization problems. It must be noted in this regard that the convergence to a binary state does not eliminate the problem of unfeasible solutions.

In this paper, we present a dynamical Ising machine that goes against the common wisdom. The machine does not necessarily converge to a binary state and, therefore, is non-binary. Yet, the machine's terminal states trivially produce binary states and do not produce unfeasible solutions for labeling problems. We demonstrate these features of the presented Ising machine for the example of finding Latin squares and solving the Sudoku puzzles, the problems that can be straightforwardly represented as coloring problems.
Besides opening a new direction in developing dynamical Ising machines, the presented machine gives a nontrivial example of how binary information can be carried by continuous dynamical systems.

\section{$V_2$ model}
\label{sec:v2-model}

Most modern dynamical Ising machines traverse a continuous spin-space as they dynamically evolve in such a way that they only settle in states where all spins are binary or binary with some marginal error \cite{Marandi2014, Leleu2017, Goto2019, Shukla2023CustomHeuristic}. 
%
For example, in~\cite{Zhang2022AMachine-BRIM, Wang2019NewMachines}, the convergence to a binary state is enforced by imposing bi-stability on the spins. 
The design choice for binary terminal states is justified by the equivalence of the designed machine's Hamiltonian in such states and the Ising Hamiltonian. At the same time, 
it is now known that a strong bi-stability can severely affect the Ising machines' performance~\cite{Erementchouk2021, Bashar23Stabilityequal}, 
which makes unclear the practical quality of solutions produced by the Ising machines


This work aims to revisit the notion of whether continuous spins, as
internal variables of Ising machines, need to converge to binary values at
all, and, consequently, to identify dynamical continuous-spin systems that
can solve inherently discrete problems without converging to binary states.

The problem native to the Ising machines is the max-cut problem: finding
partitioning a graph with the maximum total weight of edges between the
parts.
For a given connected graph $\mathcal{G} = (\set{V}, \set{E})$ with
$N = \abs{\set{V}}$ nodes and weighted adjacency matrix $\widehat{A}$, the
partition is described by a binary function (spin configuration)
$\boldsymbol{\sigma} : \set{V} \to \left\{ -1, 1 \right\}$, or, in other words, by
assigning to each graph node $i \in \set{V}$ a binary variable
$\sigma_i \in \left\{ -1, 1 \right\} $. Then the max-cut problem can be expressed
as 
\begin{equation}
    \max_{{\boldsymbol\sigma}} {C(\boldsymbol{\sigma})}\\
    =\max_{{\boldsymbol\sigma}} \left({\frac{1}{4} \sum_{i, j=1}^N A_{i,j}\left(1-\sigma_i \sigma_j\right)}\right),
    \label{eq:maxcut-formal}
\end{equation}
where $C(\boldsymbol{\sigma})$ is discrete cut function defined over a partition
of the vertices of $\mathcal{G}$. The relation with the classical Ising model is
established by presenting
$C(\boldsymbol{\sigma}) = W / 2 - H(\boldsymbol{\sigma}) / 2$, where
$W = \sum_{i,j} A_{i,j}/2$ is the total weight of the graph edges and
$H(\boldsymbol{\sigma}) = \boldsymbol{\sigma}^T \widehat{A} \boldsymbol{\sigma}/2 =
\sum_{i,j} A_{i,j} \sigma_i \sigma_j / 2$ has the meaning of the Hamiltonian of the
antiferromagnetic Ising model on graph $\mathcal{G}$.

%
%
%

To identify the dynamics governing the non-binary Ising machines, we
reformulate the max-cut problem in the spirit of the rank-2 SDP
relaxation~\cite{Williamson1994, Burer2002}: in terms of a continuous
system of variables
$\boldsymbol{\xi} = \left\{ \xi_1, \ldots, \xi_N \right\} $, and a respective
relaxed continuous cut function $C_R(\boldsymbol{\xi}) : \mathbb{R}^N \to \mathbb{R}$. The connection with the relaxed problem proposed in this work is established by first formulating
the partition delivering the max-cut as
\begin{equation}\label{eq:maxcut-intermed}
      \boldsymbol{\sigma}'
     =\argmax_{{\boldsymbol\sigma}} \left({\frac{1}{2} \sum_{i,j}A_{i,j}\Phi_I(\sigma_i  - \sigma_j)}\right),
\end{equation}
where $\Phi_I(x) = x^2/4$ is the indicator function for the edge between the
nodes $i$ and $j$ to be cut such that $\Phi_I(0)=0$ and $\Phi_I(\pm2)=1$. Letting
the spins (and consequently the cut) to assume continuous values leads to
the following class of relaxations of the max-cut problem for $\mathcal{G}$:
\begin{equation}
 \boldsymbol{\xi}' = \argmax_{\boldsymbol{\xi}}{C_R(\boldsymbol{\xi})} 
 = \argmax_{\boldsymbol{\xi}} {\sum_{i,j}A_{i,j}\Phi_R \left(\xi_i-\xi_j\right)}.
\label{eq:relaxed-maxcut}
\end{equation}
The relaxed
cut-counting function $\Phi_R(\xi)$ is a periodic function with period $4$
constrained by $\Phi_R(0)=0$ and $\Phi_R(\pm2)=1$. A few
examples of such functions are shown in Fig.~\ref{fig:Phi-plots}. We note
that, in addition to the period translations $\xi_i \to \xi_i + 4k_i$, the
relaxed cut function $C_R(\boldsymbol{\xi})$ is invariant with respect to
homogeneous translations $\xi_i \to \xi_i + \Delta\xi$.


\begin{figure}
    \centering
    \includegraphics[width=0.4\linewidth]{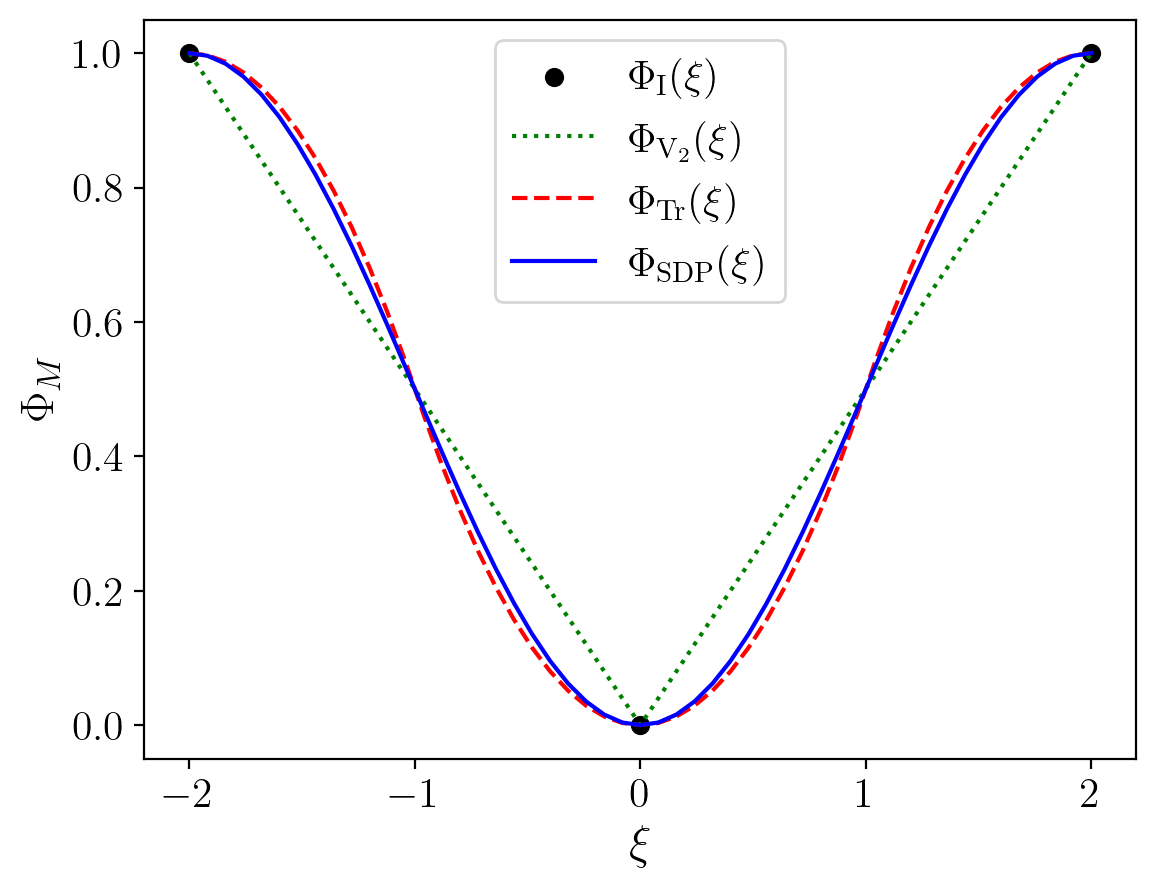}
    \caption{Examples of cut-counting functions: $\Phi_{\mathrm{I}}$ is the discrete function of the binary Ising model,
      $\Phi_{R}(\boldsymbol{\xi})$ are relaxed functions with
      $M = \mathrm{SDP}, \mathrm{Tr}, \mathrm{V}_2$ corresponding to
      rank-$2$ SDP relaxation~\cite{Burer2002}, the triangular model from
      Refs.~\cite{Shukla2024ScalableMachines, Shukla2023CustomHeuristic},
      and the $\mathrm{V}_2$ model considered in the present paper. For the
      relaxed cut-counting functions, only one period is shown.}
    \label{fig:Phi-plots}
\end{figure}

The solution to the relaxed problem \eqref{eq:relaxed-maxcut}, $\boldsymbol{\xi}'$, is not necessarily
binary. For example, 
Eq.~\eqref{eq:relaxed-maxcut} with
$\Phi_R(\xi) = \Phi_{\mathrm{SDP}}(\xi) = \left( 1 - \cos(\pi \xi/2)\right)/2$ is
equivalent to a rank-2 SDP, which outputs a binary vector only for a
special set of graphs e.g. bipartite graphs~\cite{Burer2002}. For this
reason, a rounding procedure is required to map $\boldsymbol{\xi}'$ to a
strictly binary vector. A straightforward approach is a parametric rounding
procedure~\cite{Punnen2022, Burer2002, Shukla2024ScalableMachines}, which
compares with $0$ each $\xi_i' - r$, where $r$ is the rounding center,
taken by modulo of the period of $\Phi_R(\xi)$. This defines rounding as the mapping
$\boldsymbol{\xi}'\mapsto \boldsymbol{\sigma}(r)$. Generally, this procedure produces a
spectrum of rounded configurations $\boldsymbol{\sigma}(r)$ depending on the
value of the rounding center and the respective spectrum of cuts
$C(\boldsymbol{\sigma}(r))$. Hence, a methodology to find the optimal solution
to the rounding problem must accompany solving Eq.~\eqref{eq:relaxed-maxcut}.

Next, we introduce an important relaxed spin representation~\cite{Erementchouk2023Self-containedMachines}, which 
partitions dynamical variables $\xi_i$ into two components: spin
$\sigma_i \in \pm1$
and a continuous remainder $X_i\in[-1,1)$ according to
\begin{equation}\label{eq:sigma_X-representation}
    \xi_i  = \sigma_i + X_i +  4k_i,
\end{equation}
where $k_i$ is an integer, and $4$ represents the period of $\Phi_R(\xi)$.
Thus, the vector $\boldsymbol{\xi} = \{\xi_i\}$ is replaced by the vector of
pairs $(\boldsymbol{\sigma}, \mathbf{X}) = \{\left(\sigma_i,X_i \right)\}$.

Assuming that the relaxed cut-counting function possesses the gliding
symmetry, $\Phi_R(\xi + 2) = 1 - \Phi_R(\xi)$, as all functions shown in
Fig.~\ref{fig:Phi-plots} do,
the relaxed cut function can be written
as~\cite{Erementchouk2023Self-containedMachines}
\begin{equation}\label{eq:relaxed-cut-function}
  C_R(\boldsymbol{\xi}) = C_R(\boldsymbol{\sigma}, \mathbf{X}) = C(\boldsymbol{\sigma})
  + \widetilde{C}_R(\boldsymbol{\sigma}, \mathbf{X}),
\end{equation}
where
\begin{equation}\label{eq:relaxed-cut-add}
  \widetilde{C}_R(\boldsymbol{\sigma}, \mathbf{X}) =
  \frac{1}{4} \sum_{i, j} A_{i,j} \sigma_i \sigma_j \Phi_R(X_i - X_j).
\end{equation}
Respectively, the equations of motion governing the $V_2$ Ising machine,
$\dot{\mathbf{X}} = \nabla_{\mathbf{X}} C_R(\boldsymbol{\sigma}, \mathbf{X})$, have
the form
\begin{equation}\label{eq:relaxed_X_eqm}
    \dot{X}_i = \frac{1}{2} \sum_j A_{i,j}\sigma_i \sigma_j \phi_R(X_i-X_j),
\end{equation}
where $\phi_R(X) = d\Phi_R(X) /dX$. At $X=\pm1$, the evolution of $X$ itself is discontinuous, as
illustrated by Fig.~\ref{fig:sigXspace}.  The evolution of  $\sigma$ stems from
these discontinuous transitions of $X$: whenever $X_m$
crosses the $X = \pm1$ boundary, $\sigma$ changes sign.

\begin{figure}
  \centering
  \includegraphics[width=0.7\linewidth]{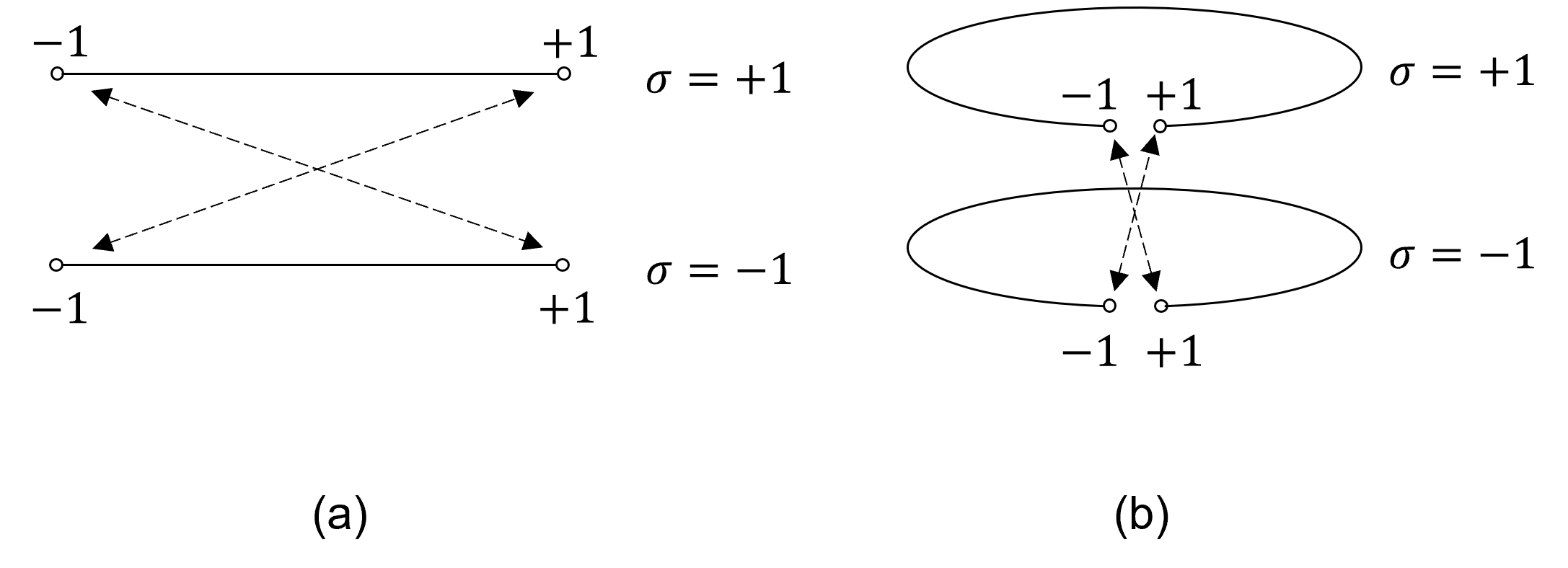}
  \caption{The phase space of the dynamical variables in the $V_2$ model
    is a wedge sum of two circles with a circumference of $2$. (a) Range of an
    state-vector element $(\sigma, X)$, with $X$ depicted as an open interval.
    (b) $(\sigma, X)$ with $X$ shown as a circle to the visualize the
    continuity of the transition between states $(\sigma,\pm1)$ and
    $(-\sigma,\mp1)$.}
  \label{fig:sigXspace}
\end{figure}

From the rounding perspective,
$\boldsymbol{\sigma}$ entering representation~\eqref{eq:sigma_X-representation}
can be regarded as rounding of $\boldsymbol{\xi}$ with respect to a particular
choice of the rounding center. 
Thus, generally,
finding the terminal states of the dynamics governed by
Eq.~\eqref{eq:relaxed_X_eqm} should be followed by solving the optimal
rounding problem. However, in~\cite{Erementchouk2023Self-containedMachines}, it was
found that this problem is eliminated for the model driven by
\begin{equation}\label{eq:phi_V-sgn}
  \phi_V(X) = \frac{1}{2}\sgn(X),
\end{equation}
where $\sgn(X)$ is the sign function. 
Owing to the shape of the plot $\Phi_{\mathrm{V}}(X) = \abs{X} /2$, the dynamics described
by $\phi_{\mathrm{V}}(\xi)$ is dubbed the $V_2$ model, and the Ising
machine based on the dynamics is called the $V_2$-machine. The $V_2$ model
possesses several key properties, which collectively summarize its behavior
as a self-contained dynamical Ising machine. Their proofs and some other
general properties of the $V_2$ model are provided
in~\cite{Erementchouk2023Self-containedMachines}.

First, as the system progresses from an arbitrary initial state, relaxed cut $C_R$ approaches the canonical discrete cut $C$ at the steady state. 
It can be shown that at any point of time,
\begin{equation}
  C_V(\boldsymbol{\sigma},\boldsymbol{X}) = C(\boldsymbol{\sigma}) +
  \frac{1}{4}\sum_{i,j}A_{i,j}\sigma_i\sigma_j \abs{X_i-X_j},
\end{equation}
and that at the steady state, $C_R(\boldsymbol{\sigma},\boldsymbol{X})$ depends on $\boldsymbol{\sigma}$ only while any dependence on $\boldsymbol{X}$ vanishes. 

Second, starting from an arbitrary initial state, the $V_2$-machine evolves so that the cut does not decrease with time.
As time progresses, the cut represented by the binary component $\boldsymbol{\sigma}$ does not decrease:
\begin{equation}
    C\left(\boldsymbol{\sigma}(t_2)\right) \geq C\left(\boldsymbol{\sigma}(t_1)\right)
\end{equation}
for $t_2>t_1$.
Equivalently, any evolution following weak perturbations from a steady state does not decrease discrete cut $C$.

Finally, it was shown in~\cite{Erementchouk2023Self-containedMachines} that
the terminal state of the $V_2$-machine does not require finding the
optimal rounding. While solving the rounding problem,
$\xi_i - r = \sigma_i(r) + X_i(r) + 4k_i(r)$, may result in different binary
configurations $\boldsymbol{\sigma}(r)$ depending on the rounding center, the
cut produced by these configurations has the same weight. Combining with
the two properties above, this implies that starting from an arbitrary
state $\boldsymbol{\xi}(0)$, the $V_2$-machine ends up in such state
$(\boldsymbol{\sigma}, \boldsymbol{X})$ that $\boldsymbol{\sigma}$ produces a cut not
worse than that obtained by the optimal rounding of $\boldsymbol{\xi}(0)$.

These properties yield dynamics with an ever-accessible solution to the
max-cut problem, whose quality may be further improved with slight but
regular perturbation. It must be emphasized that these properties are specific
for the $V_2$ model. For example, for other relaxations shown in
Fig.~\ref{fig:Phi-plots}, $\Phi_{\mathrm{SDP}}$ and $\Phi_{\mathrm{Tr}}$, the
respective discrete component $\boldsymbol{\sigma}$ does not necessarily hold
the best binary state and optimal rounding requires solving a disparate
optimization problem~\cite{Burer2002, Shukla2024ScalableMachines}.

\section{Graph coloring and stable definite color states}

The max-cut problem has the benefit of a relatively straightforward
connection with the dynamical models governing Ising machines. In other
words, it can be regarded as native to Ising machines. Consequently,
applying Ising machines to solving high-level problems requires
representing the original problem in terms of graph partitioning. From the
theoretical perspective, the sole existence of such representation is
warranted by the NP-completeness of the max-cut problem as the decision
problem about the existence of cut with the given
weight~\cite{Garey1990ComputersNP-Completeness}. However, this does not
address the practical concerns about finding such representation and making
it efficient.

One of the origins of the arising difficulties is the very 
mapping of
high-level variables of the original problem to binary spins within
the framework of interacting spin networks. In this case, the original
variables are represented by several spins (spin complexes), and, in
addition to providing good quality solutions, it is imperative that
individual spin complexes converge to states corresponding to valid values
of the high-level variables. However, \emph{a priori}, it is not evident
that a continuous-spin network governed by particular equations of motion
will converge to such states starting from generic initial conditions. For the example of the graph coloring problem considered in detail below, this is illustrated by Fig.~\ref{fig:convergence}, which shows the probability
of convergence of several Ising machines to a state describing a valid
color for a graph with a single node. For each machine, the probability is
estimated by restarting the machine $5000$ times and counting the number of
instances when the machine terminated at a state with a definite color. For
the triangular model~\cite{Shukla2024ScalableMachines},
Fig.~\ref{fig:convergence} shows the outcome of optimally rounding the
nonbinary terminal state. The coherent Ising machine (CIM)~\cite{yamamotoCoherent2017} was simulated
using package
\verb|cim_optimizer|~\cite{Chen_cim-optimizer_a_simulator_2022}. The
performance of CIM depends on the choice of hyperparameters. We evaluated
the sets of values of hyperparameters provided in the package documentation
and found that, for the considered problem, the default choice yields the
best results, which are shown in Fig.~\ref{fig:convergence}.

\begin{figure}[tb]
  \centering
  \includegraphics[width=3in]{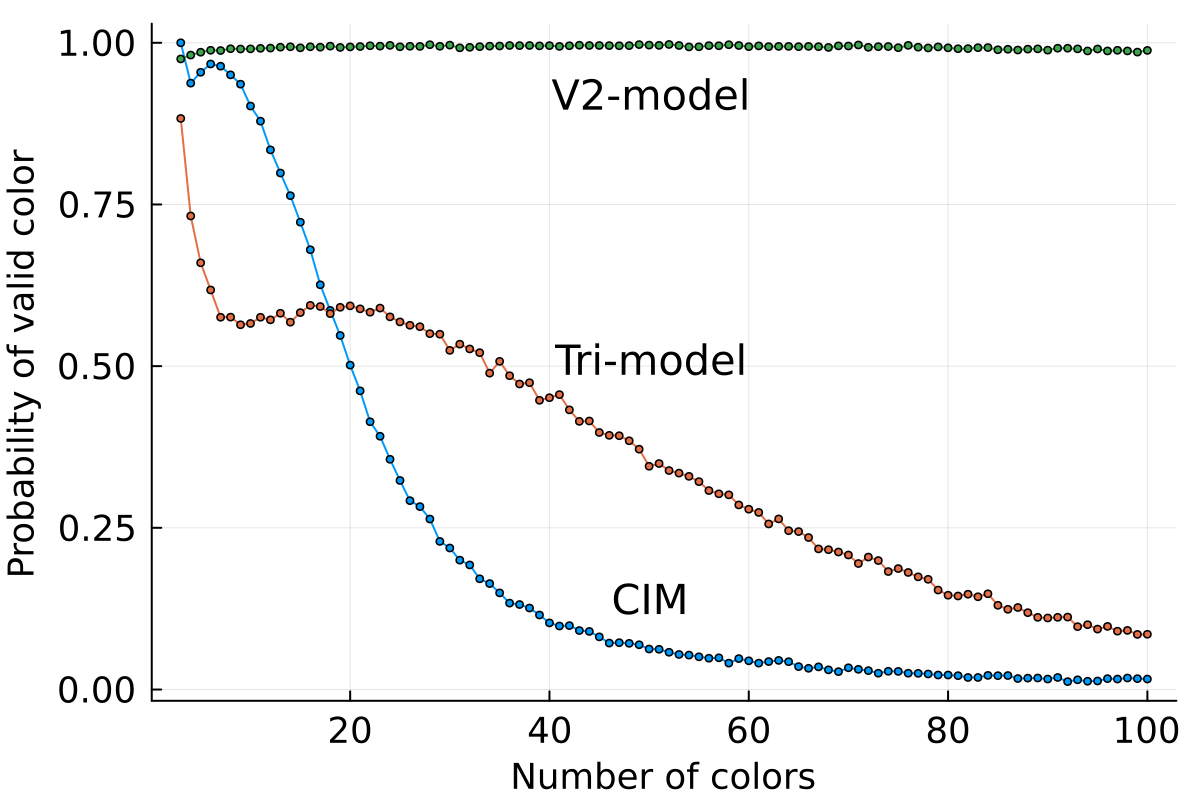}
  \caption{The dependence of the probability of convergence to a definite
    color state for a single node graph on the number of colors. Three
    machines are shown: the coherent Ising machine (CIM)~\cite{yamamotoCoherent2017} and based on the
    triangular~\cite{Shukla2024ScalableMachines} and $V_2$ models.}
  \label{fig:convergence}
\end{figure}

\subsection{Graph coloring}

A proper (node) graph coloring is an assignment of colors to the graph
nodes without adjacent nodes bearing the same color. The minimum number of
colors needed for proper coloring of graph $\mathcal{G}$ is called the graph
chromatic number, $\chi_{\mathcal{G}}$. Finding the chromatic number of a graph from a
sufficiently rich family of graphs is an NP-hard
problem~\cite{Karp1972ReducibilityProblems}.
Graph coloring may model a limited resource allocation problem, where
colors represent the shared resource, nodes play the role of the resource user, and edges denote the potential conflict arising from the simultaneous
use of the assigned resource~\cite{Hale1980FrequencyApplications,
  Boev2023Quantum-inspiredAssignment}. Real-world applications that can be
naturally formulated in terms of the graph coloring problem include
job-scheduling on machines where a set of jobs (vertices) must be scheduled
over a minimum number of time-slots (colors) while excluding the
simultaneous use of time-slots (edges)~\cite{Leighton1979AProblems}.
Another example is the problem of channel allocation in dense
wireless networks~\cite{Wu2022Coloring-BasedApproach}. In this case, nodes
stand for individual devices, colors are associated with the channels, and
edges represent the interference due to sharing the same channel.

To apply Ising machines, the $K$-coloring problem, when the graph
$\mathcal{G}$ is colored with no more than $K$ colors, must be formulated as the
max-cut problem~\cite{Lucas2014,Glover2019QuantumModels}. A straightforward
way to obtain such a formulation is to start by representing the graph
coloring by a set of binary variables $s_{i,\kappa}\in\{0,1\}$, where
$i=1,2, \ldots , N$ enumerate the graph nodes and $\kappa=1,2, \ldots , K$ correspond to
colors. Consequently, $s_{i,\kappa}$ serve as state variables:
$s_{i,\kappa} = 1$ if the $i$-th node is assigned color $\kappa$, and
$s_{i,\kappa} = 0$ otherwise. Thus, for given node $i$, the collection of binary
variables $s_{i,\kappa}$ represents a one-hot encoding of the assigned color.

The cost function characterizing the coloring comprises two components. The
first component aims to make the coloring proper: adjacent vertices must
have different colors. This condition is expressed as
$A_{i,j} \sum_\kappa s_{i,\kappa} s_{j,\kappa} = 0$ for all edges
$(i,j) \in E$. In the optimization context, this amounts to penalizing
configurations that do not meet the requirement of proper coloring by
$H^{(1)}(\boldsymbol{s}) = \sum_{i,j} A_{i,j} \sum_\kappa s_{i,\kappa} s_{j,\kappa}/2$.

Next, for each node $i$, only one of the $K$ binary variables $s_{i,\kappa}$ may
take the $+1$ value, which leads to the one-hot encoding constraint
$\sum_{\kappa=1}^K s_{i,\kappa} = 1$. This constraint can be incorporated by using the
Lagrangian relaxation when configurations violating the constraint are
assigned a penalty
$H^{(2)}(\boldsymbol{s}) = \sum_i \left(\sum_{\kappa=1}^K s_{i,\kappa} -1\right)^2$ incorporated
into the cost function with the respective Lagrange multiplier $\lambda$. Thus,
$K$-coloring is found as found by solving an optimization problem: finding
the minimum of the cost function
$H(\boldsymbol{s}) = H^{(1)}(\boldsymbol{s}) + \lambda H^{(2)}(\boldsymbol{s})$.

To solve the $K$-coloring problem using the $V_2$-machine, the problem
is reformulated in terms of the native to the model spin
variables $\sigma\in\{\pm 1\}$, using the substitution
$s_{i\kappa} = (\sigma_{i\kappa} + 1)/2$. This yields the penalty associated with
colliding colors
\begin{equation}
\label{eq:cost-comp-1}
H^{(1)}(\boldsymbol{\sigma}) = \frac{1}{2} \sum_{i,j} A_{i,j}
       \left(\sum_\kappa \sigma_{i,\kappa} \sigma_{j,\kappa} + 2\sum_\kappa \sigma_{i,\kappa}\sigma_{0} \right),
\end{equation}
and the one-hot encoding constraint, $\sum_\kappa \sigma_{i,\kappa} = 2 - K$,
\begin{equation}
\label{eq:cost-comp-2}
H^{(2)}(\boldsymbol{\sigma}) = \sum_{i, \kappa, \gamma} \sigma_{i,\kappa} \sigma_{i,\gamma} +
2(K - 2) \sum_{i, \kappa} \sigma_{i,\kappa}\sigma_{0}.
\end{equation}
In Eqs.~\eqref{eq:cost-comp-1} and~\eqref{eq:cost-comp-2}, we have
introduced an auxiliary spin $\sigma_{0} \equiv 1$, which ensures that the cost
function  is quadratic in spin variables. Collecting these
contributions together, we obtain the spin formulation of the graph
coloring problem
\begin{equation}\label{eq:color_spin}
  \boldsymbol{\sigma}' = \argmin_{\boldsymbol{\sigma}} \left( H^{(1)}(\boldsymbol{\sigma}) + \lambda
  H^{(2)}(\boldsymbol{\sigma})\right).
\end{equation}
Equations~\eqref{eq:cost-comp-1},~\eqref{eq:cost-comp-2},
and~\eqref{eq:color_spin} reveal an important structure of the formulation
of the $K$-coloring problem for the Ising machines. The coloring penalty,
$H^{(1)}(\boldsymbol{\sigma})$, after constructing the relaxation following the
procedure outlined in Section~\ref{sec:v2-model}, yields the dynamics of
spins on graph $\mathcal{G} \times \mathcal{K}_K$, a tensor product of the original graph
$\mathcal{G}$ and the complete graph $\mathcal{K}_K$ describing the one-hot encoding
constraint, with an added apex, a node connected to all other nodes. In turn,
$H^{(2)}(\boldsymbol{\sigma})$ produces the dynamics resulting in the emergence of valid
colors on $N$ disconnected complete graphs $\mathcal{K}_K$ with apexes.

For writing the equations of motion governing the $V_2$ machine, it is
convenient to introduce a homogeneous enumeration of the variables across
the auxiliary spin and the tensor product of graph nodes and color encoding
sets using $\alpha, \beta \in \left\{ 0, 1, \ldots, NK \right\}$ corresponding to
$\left\{ 0, (1,1), (1, 2), \ldots , (N, K) \right\}$. Using this enumeration,
the penalty terms $H^{(k)}(\boldsymbol{\sigma})$, with $k = 1, 2$, can be written as
quadratic forms,
$H^{(k)}(\boldsymbol{\sigma}) = \boldsymbol{\sigma}^T \widehat{A}^{(k)} \boldsymbol{\sigma}$,
with the adjacency matrices
\begin{equation}\label{eq:k-color_adj-mats}
  \widehat{A}^{(1)} = \mqty(0 & K \mathbf{e}_{NK} \\
  K \mathbf{e}^T_{NK} & \widehat{A}_{\mathcal{G}} \times \widehat{I}_K),
  \qquad
  \widehat{A}^{(2)} = \mqty(0 & (K - 2) \mathbf{e}_{NK} \\
  (K - 2) \mathbf{e}^T_{NK} & \widehat{I}_{N} \times \widehat{A}_{\mathcal{K}}),  
\end{equation}
where $\mathbf{e}_{NK} = \left( 1, \ldots, 1 \right)$ is a vector with $NK$ unit
components, $\widehat{I}_M$ is the $M \times M$ identity matrix, $\widehat{A}_{\mathcal{G}}$ and
$\widehat{A}_{\mathcal{K}}$ are the adjacency matrices of graphs $\mathcal{G}$ and $\mathcal{K}_K$,
respectively, and $\widehat{A} \times \widehat{B}$ denotes the Kronecker product
of matrices $\widehat{A}$ and $\widehat{B}$.

The respective relaxed forms of the penalty terms defining the dynamics of
the $V_2$ machine are $H^{(k)}_V(\boldsymbol{\sigma}, \mathbf{X}) =
H^{(k)}(\boldsymbol{\sigma}) + \widetilde{H}^{(k)}_V(\boldsymbol{\sigma},
\mathbf{X})$, where
\begin{equation}\label{eq:V_H_adds}
  \widetilde{H}^{(k)}_V(\boldsymbol{\sigma}, \mathbf{X}) =
  \frac{1}{4} \sum_{\alpha,\beta} A^{(k)}_{\alpha, \beta} \sigma_\alpha \sigma_\beta \abs{X_\alpha - X_\beta}.
\end{equation}

The separated dynamics structure is
reflected by the respective equations of motion 
\begin{equation}\label{eq:k-color_eqm}
  \dot{\mathbf{X}} = -\nabla_{\mathbf{X}} \widetilde{H}^{(1)}_V(\boldsymbol{\sigma}, \mathbf{X})
  - \lambda \nabla_{\mathbf{X}} \widetilde{H}^{(2)}_V(\boldsymbol{\sigma}, \mathbf{X}),
\end{equation}
where the negative sign denotes that the machine evolves towards the
minimum of the penalty, and the coordinate of the auxiliary spin is kept
fixed $\dot{X}_0 \equiv 0$.

\subsection{Emergence of definite color states}


The separated form of equations of motion~\eqref{eq:k-color_eqm} allows us
to pose the problem of the emergence of the states corresponding to definite colors. 
The dynamics of individual nodes toward the definite color states is
governed by $H^{(2)}_V(\boldsymbol{\sigma}, \mathbf{X})$, which is independent for different
nodes. Therefore, to investigate the emergence of definite colors, it
suffices to consider this dynamics for graph $\mathcal{K}_{K+1}$ corresponding to
coloring of a one-node graph. In this case, the dominant factor determining
the dynamics of the spin variables in the spin complex is their interaction
with the auxiliary spin owing to the relatively large weight of the
respective coupling. Considering the dynamics represented on a circle with
circumference $2$, generally, ``free'' spins co-oriented with the auxiliary
spin are repelled from it, while counter-oriented ones are attracted, as
illustrated by Fig.~\ref{fig:start-end-spins-labeled}.

\begin{figure}[tb]
  \centering
  \includegraphics[width=3.2in]{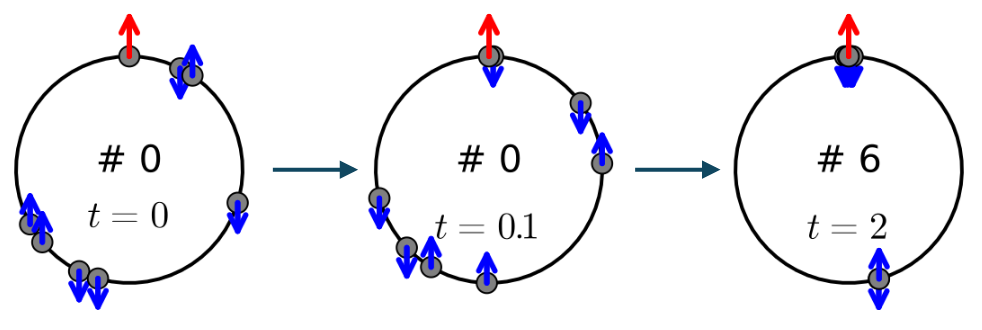}
  \caption{An example of the progression of a spin complex representing
    $K = 7$ colors from an initial generic state to a state with definite
    color. The stationary red vertical arrow shows the auxiliary spin. The
    label at the center shows the identified color ($0$ corresponds to the
    state that does not represent a valid color). Shown time is measured in
    units of the equations of motion [Eq.~\eqref{eq:k-color_eqm}].}
  \label{fig:start-end-spins-labeled}
\end{figure}


\begin{theorem}\label{thm:stable_states}
  The only stable equilibrium states of the $V_2$ dynamics of the spin
  complexes representing a color are states with definite color.
\end{theorem}

\begin{proof}
  The theorem statement follows from the general property of the $V_2$
  model established in~\cite{Erementchouk2023Self-containedMachines}: for
  any weighted graph $\mathcal{G}$, only states
  $(\boldsymbol{\sigma}, \mathbf{X})$ with $\boldsymbol{\sigma}$ yielding the maximum
  cut partition of $\mathcal{G}$ are stable equilibria.

  In turn, inspecting the definition of $H^{(2)}(\boldsymbol{\sigma})$, one can
  see that the definite color states correspond to the maximum cut of the
  graph with the adjacency matrix given by $\widehat{A}^{(2)}$. It is,
  however, constructive to show this directly. We notice that graph
  $\mathcal{K}_{K+1}$ described by the adjacency matrix $\widehat{A}^{(2)}$ is
  obtained from $\mathcal{K}_K$ with edges of unit weight by adding an apex (the
  auxiliary spin) connected
  with the $\mathcal{K}_K$ by edges with weight $K-2$. Thus, graph
  $\mathcal{K}_{K+1}$ is symmetric with respect to permutations of the nodes of the
  ``underlying'' graph $\mathcal{K}_K$. Therefore, the weight of the cut of graph $\mathcal{K}_{K+1}$ is determined only by the number of positive and negative spins
  on graph $\mathcal{K}_K$.

  Let the part of spin configuration $\boldsymbol{\sigma}$ have of
  $\mathcal{K}_K$ spins $U \geq 0$ and $D \geq 0$ number of spins equal to
  $+1$ and $-1$, respectively ($U + D = K$). Then, the cut of graph   $\mathcal{K}_{K+1}$ produced by $\boldsymbol{\sigma}$ has the weight
  \begin{equation}\label{eq:cut_Kp1}
    C_{K+1}(\boldsymbol{\sigma}) = D(K - 2) + UD,
  \end{equation}
  where the first term is the weight of cut apex edges, and the second term
  is the weight of cut of the $\mathcal{K}_{K}$ edges. It is straightforward to check
  that $C_{K+1}(\boldsymbol{\sigma})$ reaches maximum
  $\overline{C}_{K+1} = (K - 1)^2$ at $U = 1$, that is a definite color
  state.
\end{proof}


Identifying stable equilibria is not sufficient for describing the terminal
states of the $V_2$ machine, as, starting from a generic initial state,
the machine may terminate in a saddle critical point of $H^{(2)}_V
(\boldsymbol{\sigma} , \mathbf{X})$. To explain the high probability of
convergence demonstrated in Fig.~\ref{fig:start-end-spins-labeled}, we need
more refined description of the equilibria.

We take into account that, as shown
in~\cite{Erementchouk2023Self-containedMachines}, equilibria of the $V_2$-machine on any weighted graph
$\mathcal{G}_W = \left\{ \set{V}_W, \set{E}_W \right\}$ have the form of strong
clustered states $\left( \boldsymbol{\sigma}, \mathbf{X} \right) $, with the set
of nodes $\set{V}_W$ partitioned into $R(\boldsymbol{\sigma})$ (the order of the
state) disjoint subsets according to $\set{V}_W = \bigcup_p \set{V}^{(p)}$, where
$\set{V}^{(p)} = \left\{ v \in \set{V}_W : X_u = X^{(p)} \right\}$ and
$1 \leq p \leq R(\boldsymbol{\sigma})$. As was shown
in~\cite{Erementchouk2023Self-containedMachines}, the state order is
bounded from above by the degeneracy of the binary configuration
$\boldsymbol{\sigma}$ defined as the number of graph partitions with the same
cut weight as produced by
$\boldsymbol{\sigma}$.


To consider perturbed equilibria, we need to
extend the notion of clustered states beyond strong clusters. To this end,
we assume that the set $\left\{ X^{(p)} \right\}$ characterizing the
equilibrium state is covered by nonoverlapping \emph{open} intervals
$\left(X^{(p)}_-, X^{(p)}_+\right)$ such that the set of nodes is
partitioned between
$\set{V}^{(p)} = \left\{ v \in \set{V}_W : X_u \in \left( X^{(p)}_-, X^{(p)}_+
  \right) \right\}$ and all $\set{V}^{(p)}$ are non-empty.

Due to the symmetry discussed in the proof of
Theorem~\ref{thm:stable_states}, the matrix elements of $\widehat{A}^{(2)}$
can be presented for $\alpha \ne 0$ as $A^{(2)}_{\alpha,\beta} = q_\beta$ with $q_0 = K - 2$
and $q_\beta = 1$ for $\beta > 0$. Using this
representation and the weak cluster partitioning, the equations of motion
can be written in a form
\begin{equation}\label{eq:one-node-h_2-dyn}
  \dot{X}_\alpha = \frac{1}{2}\sum_{\beta \in \set{V}^{(p)}} q_{\beta} \sigma_\alpha \sigma_\beta \sgn \left( X_\alpha - X_\beta \right) 
  + \frac{\sigma_\alpha \overline{Q}^{(p)}}{2},
\end{equation}
where we have introduced
$\overline{Q}^{(p)} = \sum_{q < p} Q^{(q)} - \sum_{q > p} Q^{(q)}$, a charge
external to cluster $\set{V}^{(p)}$ determined by the charge associated
with the clusters
\begin{equation}\label{eq:cluster-charge}
  Q^{(q)} = \sum_{\beta \in \set{V}^{(q)}} q_\beta \sigma_\beta.
\end{equation}
We note that, as follows from Theorem~\ref{thm:stable_states},
the only stable equilibria of spin complexes representing colors are the
equilibria with clusters of zero total charge.

It is evident that the necessary condition for the (strong) clustered state
characterized by $\left\{ X^{(p)} \right\}$ to be an equilibrium is for the
external charge to vanish: $\overline{Q}^{(p)} = 0$, for all $p$. This gives a simple criterion for describing all equilibrium states of spin complexes
representing a color.

\begin{theorem}\label{thm:equilibria}
  Equilibrium states of the $V_2$-machine driven by
  $H^{(2)}_V(\boldsymbol{\sigma}, \mathbf{X})$ have orders $1 \leq R \leq 3$. 
\end{theorem}

\begin{proof}
  The case $R = 1$, when $X_\alpha = 0$ for $\alpha = 0, \ldots, K$, is trivial as
  the right-hand-side of Eq.~\eqref{eq:one-node-h_2-dyn} vanishes
  identically. 

  For $R > 1$, the condition of vanishing external charges,
  $\overline{Q}^{(p)} = 0$, for all $p$, has the form of a system of
  homogeneous equations with respect to cluster charges, or, in other
  words, an equation for eigenvectors corresponding to the zero eigenvalue
  of an antisymmetric $R \times R$ matrix with all the elements below the main
  diagonal equal to $1$. Subtracting the $(p - 1)$-th row, with $p > 1$, of
  the matrix from the $p$-th row, we obtain the relation between the
  charges of successive clusters: $Q^{(p)} = (-1)^p Q^{(1)}$. Adding the
  first and the last rows yields the ``boundary condition''
  $Q^{(R)} = Q^{(1)}$.

  It follows then that for the even orders, the system with respect to
  cluster charges admits only trivial solutions, $Q^{(p)} = 0$. At the same
  time, the cluster containing the auxiliary spin may have zero charge only
  when it contains $K - 2$ spins with values $-1$. Thus, there may be only
  two clusters, and, therefore, the only equilibrium states of even order
  are with $R = 2$ with the same structure. All but two relaxed spin
  variables have (for $\alpha > 0$) $\sigma_\alpha = -1$ and
  $X_\alpha = 0$. The remaining two variables form a separate cluster with
  $\sigma_\beta = -\sigma_{\beta'} = 1$ and
  $X_\beta = X_{\beta'} \in [-1, 1)\setminus \left\{ 0 \right\}$. Index
  $\beta$ yields the color represented by the spin complex. Together with the
  case $X_\beta = 0$ (that is an $R = 1$ state), this gives an explicit
  construction of all stable equilibrium states from
  Theorem~\ref{thm:stable_states}.

  For equilibrium states of odd order with $R > 1$, we only need to
  consider non-trivial solutions to $\overline{Q}^{(p)} = 0$, for all $p$.
  The strongest constraint imposed on these solutions is that the magnitude
  of cluster charges must be equal to the charge of the cluster containing
  the auxiliary spin. A simple counting argument shows that this constraint
  may only be satisfied with no more than three clusters. This leaves only
  possible $R = 3$ for equilibrium states. All such states have the same
  structure: all but four relaxed spin variables have (for $\alpha > 0$)
  $\sigma_\alpha = -1$ and $X_\alpha = 0$ yielding the cluster with a charge equal to the
  magnitude of the total state charge, $\abs{\sum_{p} Q^{(p)}} = 2$. The
  remaining four variables are split between two clusters,
  $\left\{ \beta_1, \beta_1' \right\}$ and
  $\left\{ \beta_2, \beta_2' \right\}$, with
  $\sigma_{\beta_k} = \sigma_{\beta_{k'}}$ and
  $X_{\beta_k} = X_{\beta_{k'}}$. The signs of the spins depend on the mutual
  ordering between $0$, $X_{\beta_1}$ and $X_{\beta_2}$. Such states exist
  for problems with $K \geq 4$.
\end{proof}

The structure of charged equilibrium states (with $R = 1$ or $R = 3$) has a
characteristic feature: they contain high-charge ($Q > 1$) sub-clusters
with coinciding $X$ coordinates of the constituting variables. Considering
the dynamics of two close co-oriented relaxed spins using the weak cluster
partitioning and Eqs.~\eqref{eq:one-node-h_2-dyn}, one can see that such
sub-clusters cannot form dynamically and must be present in the
initial conditions. Therefore, it can be concluded that almost all generic
initial states terminate in a stable equilibrium state.

It must be noted, however, that these results hold for formal equations of
motion. In particular implementations employing various approximations for
solving the equations of motion, there may be specific factors impacting
the convergence of the machine. For example, in the present paper, in
particular, to obtain results presented in Fig.~\ref{fig:convergence}, we
have used the Euler approximation. It is exact outside of the hyperplanes
where two or more relaxed spin variables have the same $X$ coordinates.
However, when the system phase point traverses these hyperplanes, the Euler
approximation leads to various spurious oscillations. The oscillations
affecting the convergence to definite color states occur when two
co-oriented relaxed spins change signs within the same time step. It can be
seen that such a situation may occur when the two-node cluster in an $R = 2$
equilibrium state is formed within the interval $\Delta X = \Delta t$ containing the
$X = \pm1$ point, where $\Delta t$ is the approximation time-step. Since the
position of the two-node cluster is not restricted by general dynamical
properties, the probability of such an event can be estimated as
$\sim \Delta t/2$. In turn, due to the general constraint $\Delta t \ll 2$, the
probability of encountering such spurious states is small, which is
reflected in Fig.~\ref{fig:convergence}. There are different techniques
for breaking the formation of spurious oscillations near the $X = \pm 1$ point
while generally staying within the Euler approximation. A detailed analysis
of improved algorithms for solving equations of motion governing the
$V_2$-machine [Eq.~\eqref{eq:relaxed_X_eqm} with $\phi_V = \phi_R$] is beyond the
scope of the present paper and will be provided elsewhere.

The last that needs commenting on is how the results discussed above
manifest in the context of solving of the graph coloring problem. The
relation between the coloring optimization and the convergence to the
definite color in the equations of motion is governed by the Lagrange
multiplier. Generally, $\lambda$ needs to be carefully chosen so that each
component is not too low while the other is too high, creating an
unfavorable bias towards one of the requirements imposed on the machine
terminal state. It is, therefore, important that the transition from an
equilibrium to a non-equilibrium state creates a ``restoring force'' (the
variation of the contribution of the second term in
Eq.~\eqref{eq:k-color_eqm}) with the magnitude of at least $\lambda$. Hence, if
the magnitude of the perturbation in Eq.~\eqref{eq:k-color_eqm} due to the
coloring term does not exceed $\lambda$, the machine's behavior is governed by
Theorems~\ref{thm:stable_states} and~\ref{thm:equilibria}. Since for the
problems regarded here, the perturbation vanishes when the coloring is
found, taking $\lambda = 1$, as was done in our studies discussed below, is
sufficient. A more complete description of the strategies for choosing the
correct value of $\lambda$ is a subject of ongoing research.

\subsection{Example of a proper graph coloring using the $V_2$ model}

Next, we illustrate the solution of graph coloring problem using the $V_2$
model. Figures~\ref{fig:color_example_graphs}a
and~\ref{fig:color_example_graphs}b show the example graph $\mathcal{G}$ with
$\chi_{\mathcal{G}}=3$ chosen for the illustration and an example of its proper
coloring. As described above, the problem graph of
Fig.~\ref{fig:color_example_graphs}a is transformed into a max-cut problem
graph (or an Ising graph) that has the form of
$\mathcal{G} \times \mathcal{K}_3$ supplied with an apex as shown in Fig.~\ref{fig:color_example_graphs}c.
Each node of the Ising graph is labeled as
$\sigma_{i,\kappa}$, where $i$ is the node index in the original graph and $\kappa$ is the
index of the color. The weight of the edge between nodes $\sigma_{i,\kappa}$ and
$\sigma_{j,\kappa}$
\begin{equation}
    w\left(\sigma_{i,\kappa},\sigma_{j,\kappa}\right) = 1,
\end{equation}
and between $\sigma_{i,\kappa}$ and $\sigma_{i,\mu}$
\begin{equation}
    w\left(\sigma_{i,\kappa},\sigma_{i,\mu}\right) = \lambda.
\end{equation}
The auxiliary node connects to all nodes of the Ising graph with a uniform weight
\begin{equation}
    w_{0, (i, \kappa)} = \lambda + 3.
\end{equation}

\begin{figure}
    \centering
    \includegraphics[width=0.8\linewidth]{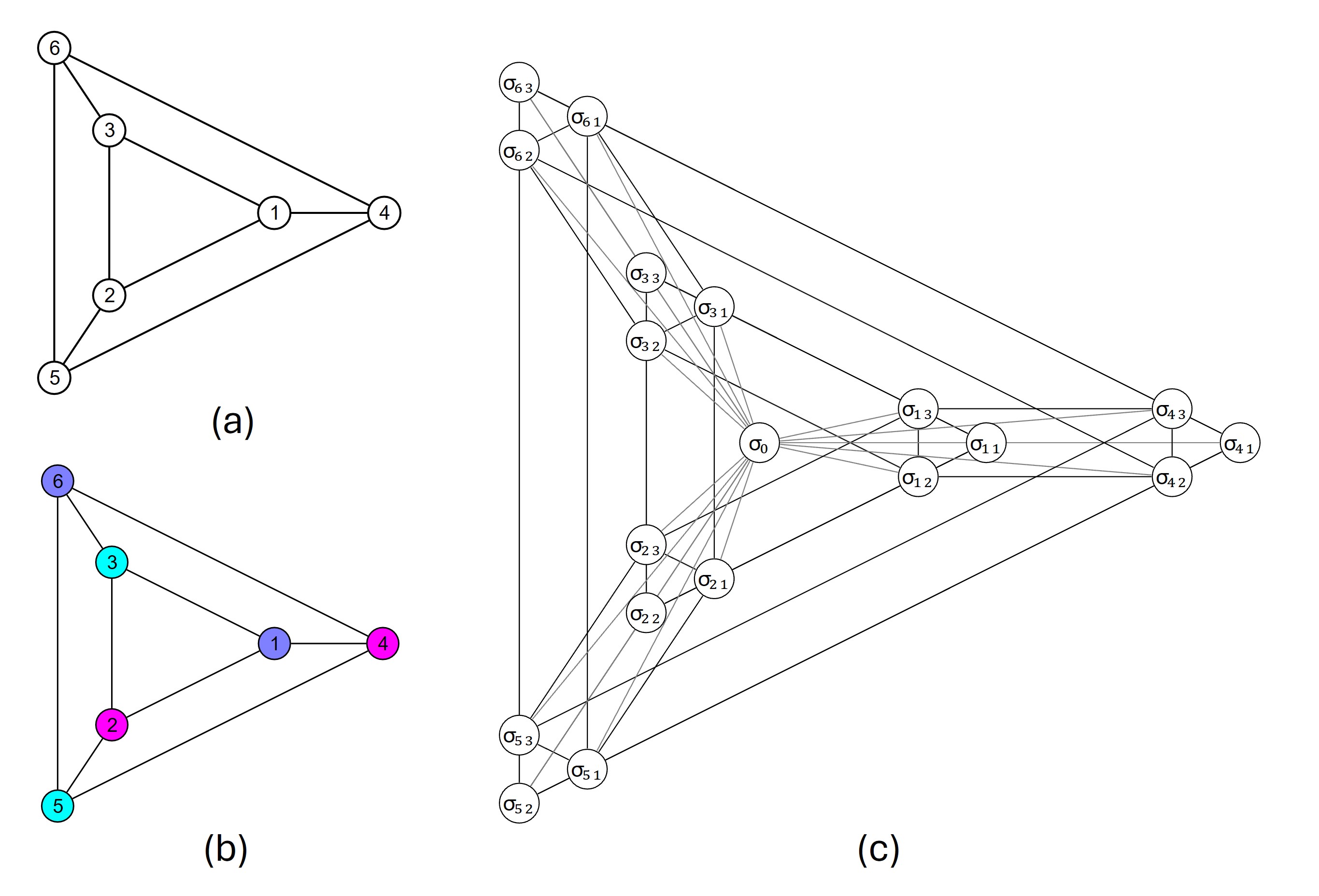}
    \caption{Example coloring problem. (a) 6-node problem graph; (b) possible coloring of the 6-node graph; (c) Ising model of the graph, the central node in the graph is the auxiliary node with the fixed spin.}
    \label{fig:color_example_graphs}
\end{figure}

Figure~\ref{fig:color_example_plots}a plots the evolution of
$\boldsymbol{\xi}$ with time and demonstrates that the $V_2$ machine solves
the coloring problem without converging to a binary state. As pointed out
before, this is in contrast with the existing machines such as
in~\cite{Inagaki2016, Goto2019}.

\begin{figure}
    \centering
    \includegraphics[width=1.0\linewidth]{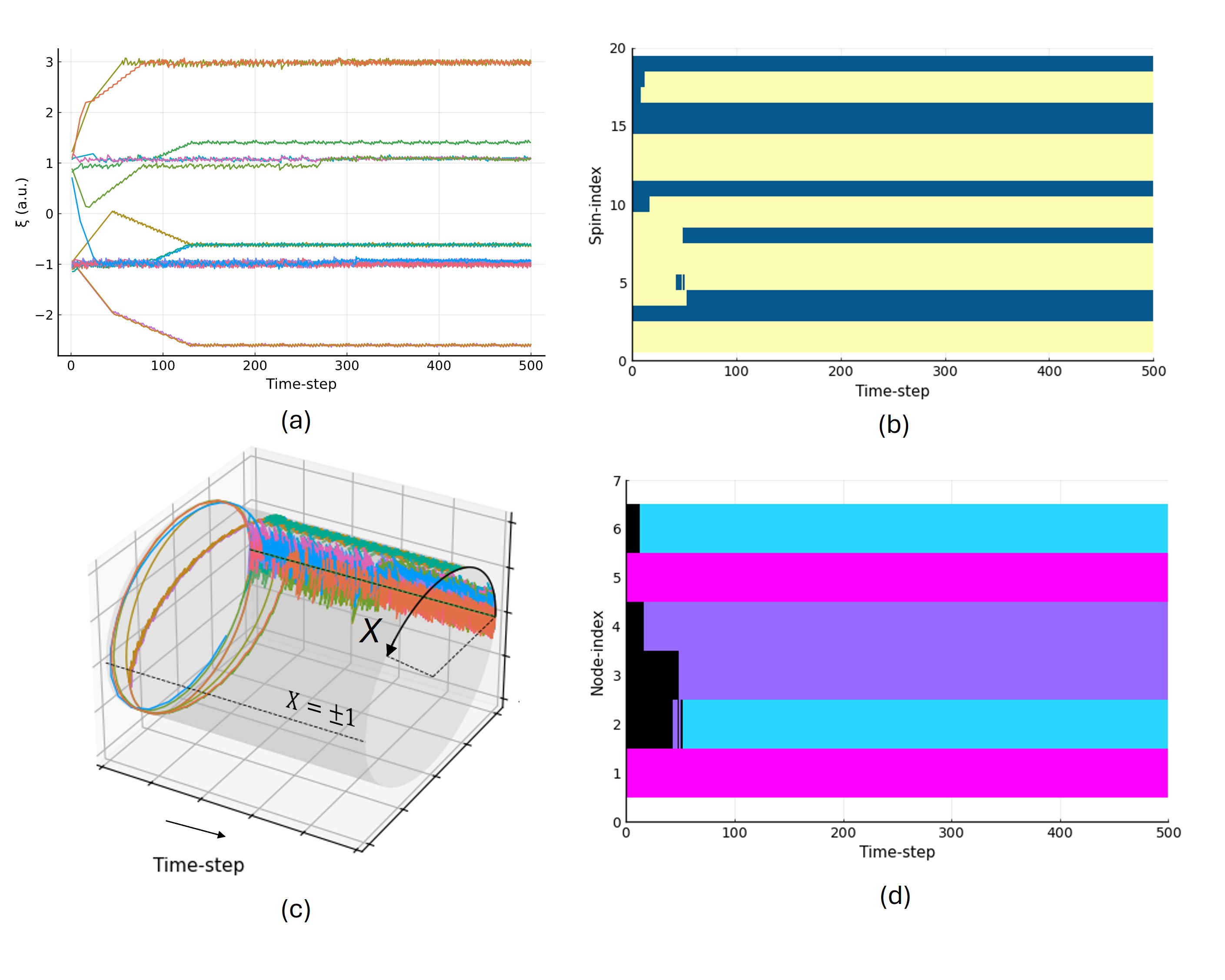}
    \caption{Dynamical progression for the example coloring problem. (a) $\xi$-evolution in $V_2$ model; (b) spin-evolution (green implying spins state of `1'); (c) evolution of the continuous component, $X$; (d) Color assignment based on the spins. Black regions indicate spin configurations that do not correspond to any color as they violate the one-hot color encoding constraint.}
    \label{fig:color_example_plots}
\end{figure}

The evolution of the relaxed spins (the $(\sigma, X)$ representation) is
depicted in Figs.~\ref{fig:color_example_plots}b
and~\ref{fig:color_example_plots}c, with the $X$ coordinates plotted on a
cylinder $[-1, 1) \times \text{time}$ with the apparent formation of the clusters
characteristic for equilibrium states of the $V_2$ dynamics. The consequence of the finite time-step of the employed Euler approximation is
the spurious oscillations near the positions of the (true) strong clusters,
similar to the oscillations discussed above.

%

The evolution of color, derived from the spin evolution, is plotted in Fig.~\ref{fig:color_example_plots}d. Black strips indicate the absence of valid color assignment for that node.
Thus, we see that the $V_2$ model successfully solves the coloring problem of the selected graph by providing sufficiently good solutions to the corresponding max-cut problem, especially so that each vertex is assigned one (and only one) color.
In contrast, the coloring of slightly larger planar graphs on CIM was shown to require developing a special multi-stage procedure involving a regular adjustment of adjacency matrix~\cite{Inaba2022}.

\section{Latin squares and the Sudoku puzzle}

\begin{figure}
    \centering
    \includegraphics[width=1.0\linewidth]{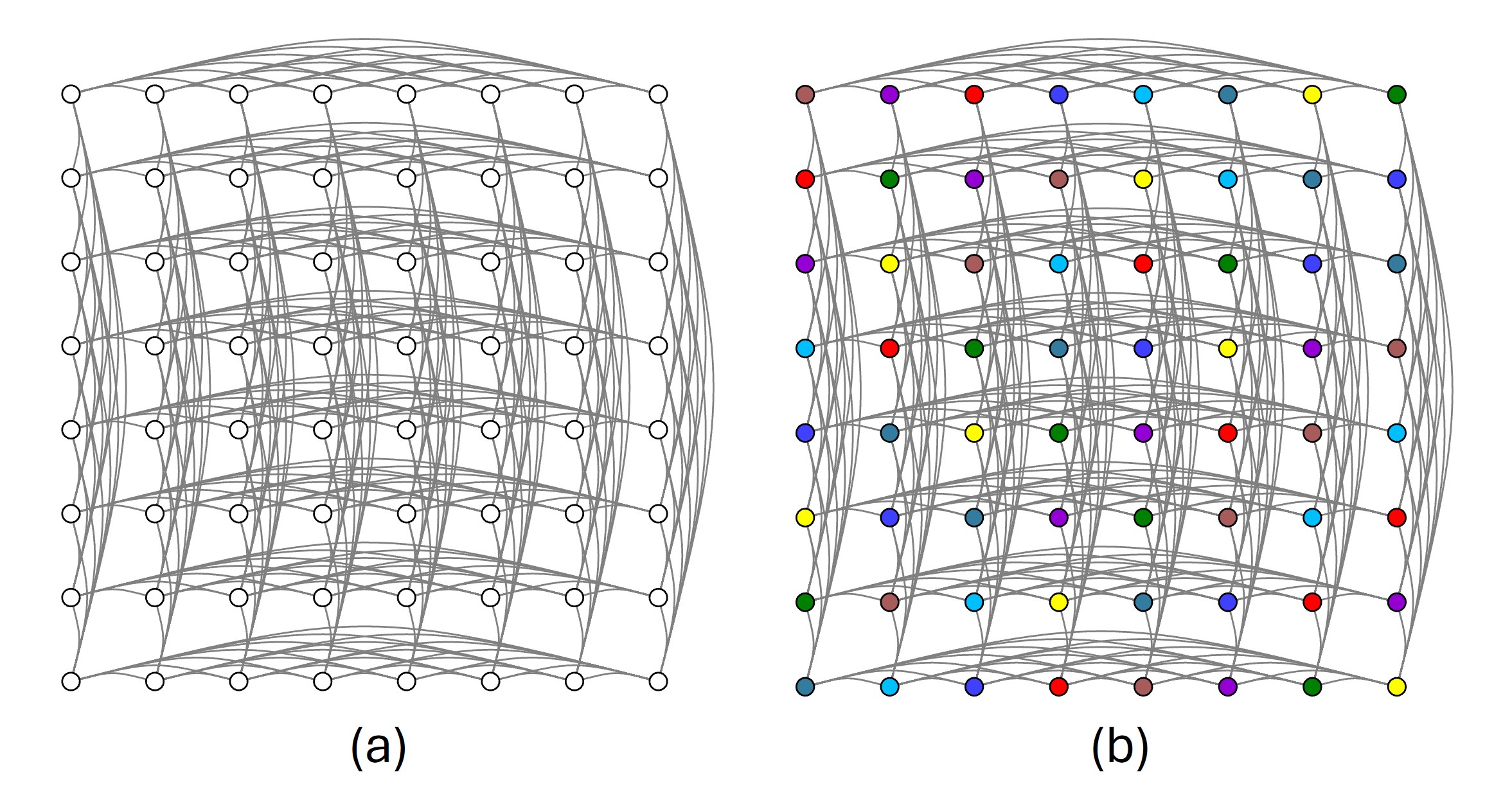}
    \caption{Latin square/rook's graph. (a) 8$\times$8 uncolored; (b) 8$\times$8 colored graph}
    \label{fig:latinsq_graph}
\end{figure}


A Latin square is an $N\times N$ (square) array of $N$ distinct labels arranged so that each row and column of the square has all the $N$ labels exactly once, without repetitions or omissions. The number of such arrangements is known to grow super-exponentially with $N$~\cite{VanLint1992ACombinatorics}, although an expression allowing for an efficient evaluation of the number of Latin squares is unknown~\cite{Stones2010TheRectangles}.
Latin squares are primarily used in statistical designs (for example, maximizing soil utility in agriculture), error-correcting codes, and potentially, at larger scales, cryptography~\cite{2006HandbookDesigns}.




Constructing an $N\times N$ Latin square is equivalent to the $N$-coloring of a graph whose vertices correspond to cells of the square and edges connect vertices belonging to the same row or column, or the size-$N$ rook's graph (Fig. \ref{fig:color_latinsq2}a) \cite{Laskar1999ChessboardParameters}.
As the most basic test of the applicability of $V_2$ model in larger scale combinatorics, we apply it to the coloring of the $8\times8$ rook’s graph. 
 
\begin{figure}
    \centering
    \includegraphics[width=1.0\linewidth]{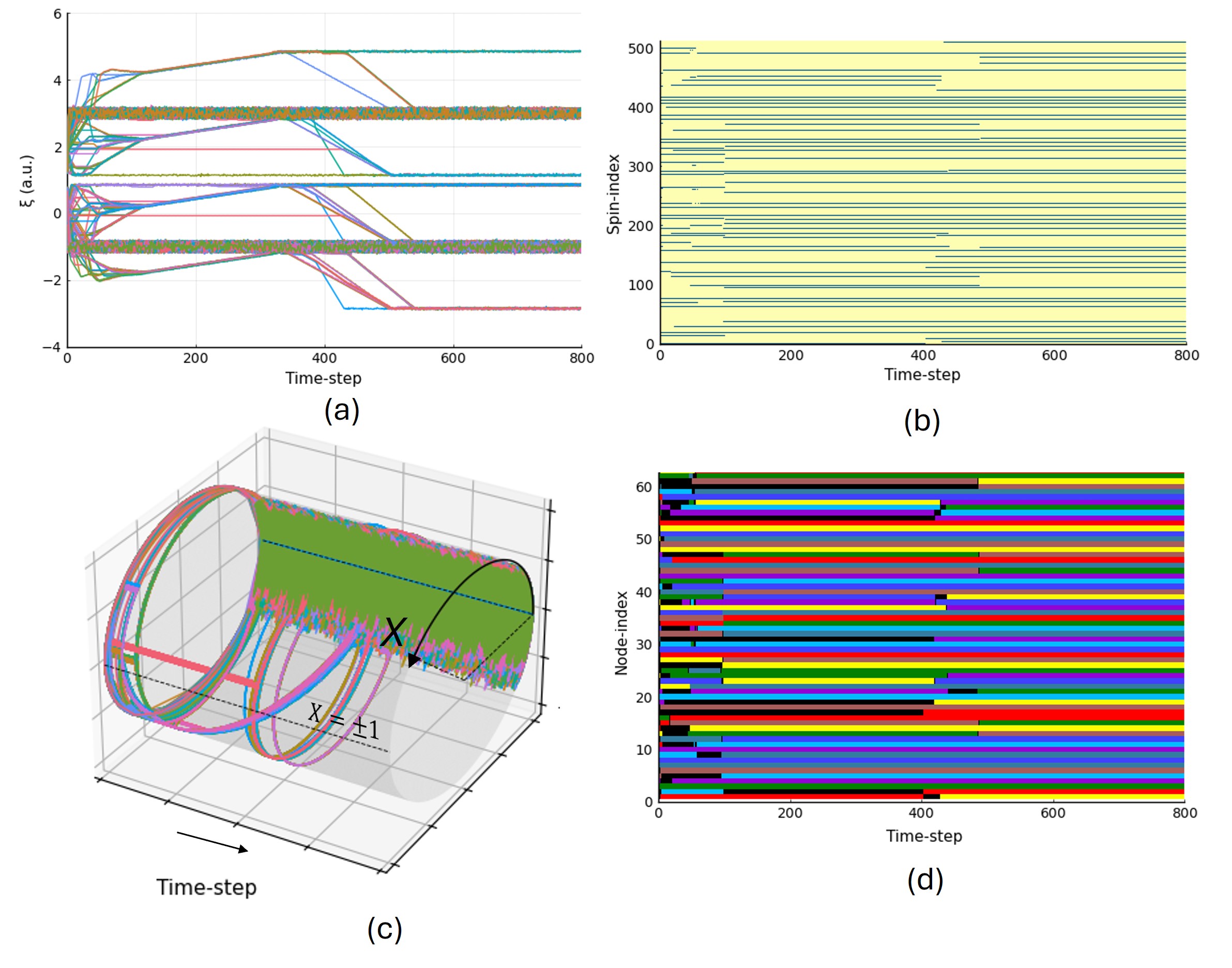}
    \caption{Dynamical progression for the coloring of rook's graph. (a) $\xi$-evolution in $V_2$ dynamics; (b) spin-evolution (green implying spins state of `1'); (c) $X$-evolution; (d) Color assignment based on the spins. Black indicates the intervals when the machine state does not yield a valid color.}
    \label{fig:color_latinsq2}
\end{figure}

Figure~\ref{fig:color_latinsq2}a shows the evolution of $\boldsymbol{\xi}$ for the Ising model of the rook’s graph using the $V_2$ dynamics when $\boldsymbol{\xi}$ is randomly initialized. 
Most features of its evolution are the same as for the example coloring problem.
In addition, we see an initial flux in the machine state, with relatively larger increments arising from a generally larger density of the graph. Throughout this flux, a well-defined, albeit suboptimal, binary state vector is always available. The evolution of $\boldsymbol{\xi}$ reaches a “noisy” steady state composed of clusters whose thickness depends on the time-discretization parameter and the typical degree of a node, as discussed above. 

Figures~\ref{fig:color_latinsq2}b and \ref{fig:color_latinsq2}c show the evolution of the spin state, $\boldsymbol{\sigma}(t)$ and the continuous component $\boldsymbol{X}(t)$, respectively.
Here, the period of $\boldsymbol{\xi}$ flux $\left(T<100\right)$ translates into the similar period of $\boldsymbol{X}$-flux where it evolves the fastest and actively traverses its range, and
terminates such that all $X$ settle in a cluster centered at $0$.

Progression of the color assignment from the $V_2$ model is shown in Fig. \ref{fig:color_latinsq2}d, where black regions correspond to the time intervals when the spin component $\boldsymbol{\sigma}$ violates the one-hot encoding constraint, that is when no valid color can be assigned. The result of the coloring of the rook’s graph is shown in Fig.~\ref{fig:latinsq_graph}b.

\begin{figure}
    \centering
    \includegraphics[width=1.0\linewidth]{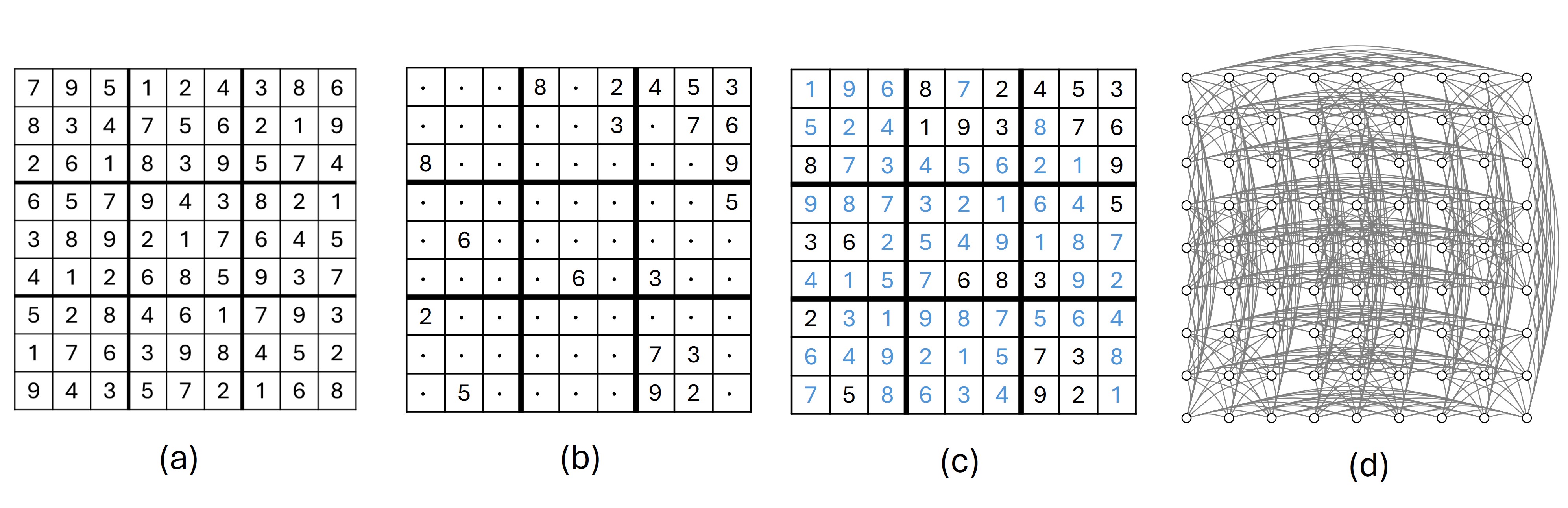}
    \caption{Elements of the Sudoku puzzle. (a) $9 \times 9$ Sudoku grid with a valid arrangement of the digits 1 to 9; (b) an example of the Sudoku puzzle; (c) solved Sudoku puzzle; (d) $81$-node Sudoku graph.}
    \label{fig:color_sudoku}
\end{figure}

Next, we apply the $V_2$ model to solve the popular combinatorics puzzle Sudoku.
This numeric puzzle can be regarded as set on a Latin square (Fig. \ref{fig:color_sudoku}a) but with additional constraints. Besides the prohibited repetitions of labels within individual rows and columns, the labels must be unique within predefined regions on the square. The canonical Sudoku puzzle is set on a $m^2 \times m^2$ grid with the predefined regions being $m \times m$ sub-squares. Finally, the puzzle is formulated by providing clues: pre-filling some cells with labels (Fig. \ref{fig:color_sudoku}b-c). For a $9\times9$ puzzle to be well-posed, that is to have a unique solution, at least 17 clues must be provided.
%
%
%
Similar to how a Latin square was found in the previous section, the Sudoku grid can be filled by coloring the Sudoku graph depicted in Fig. \ref{fig:color_sudoku}d, whereas solving Sudoku puzzles requires coloring partially colored graphs.


 \begin{figure}
     \centering
     \includegraphics[width=1.0\linewidth]{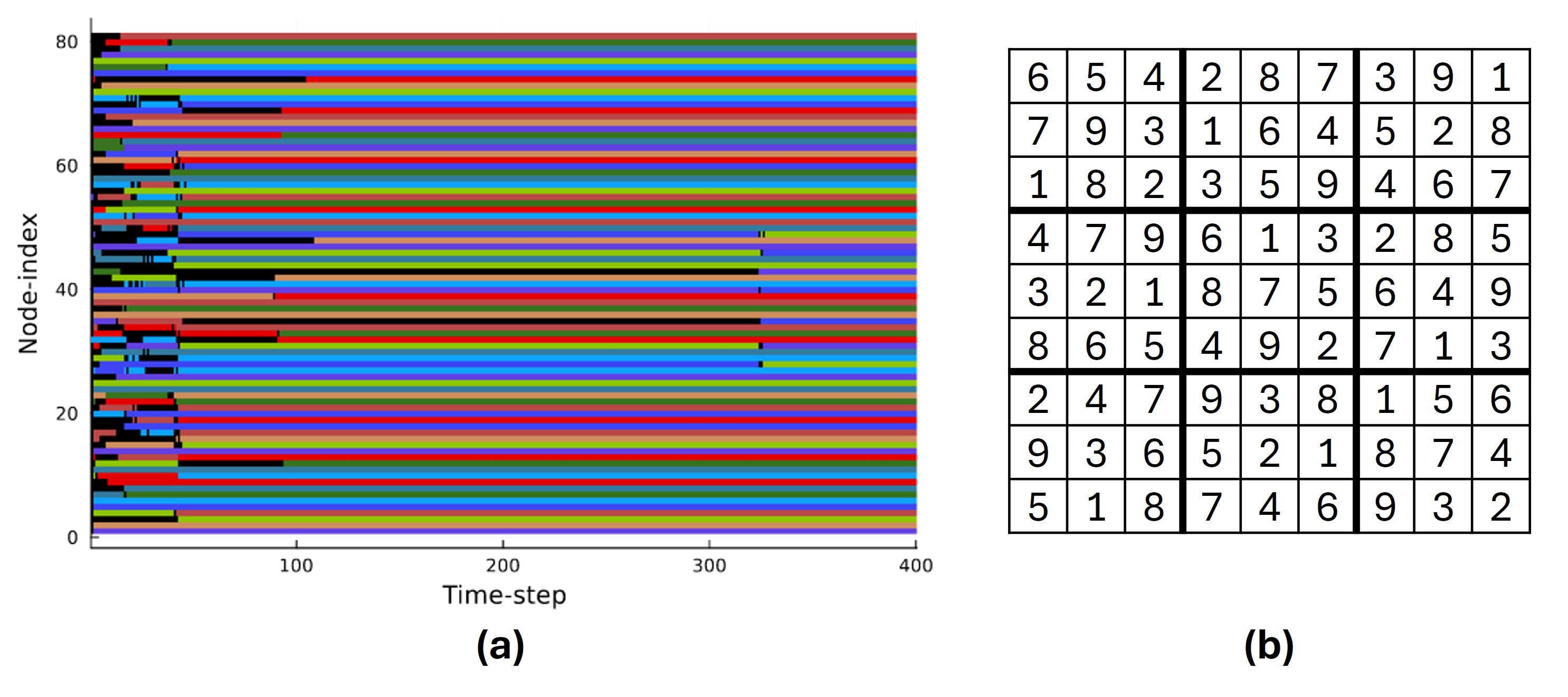}
     \caption{Filling up the Sudoku grid using the $V_2$-machine. (a) Evolution of color assigned to the individual cells; (b) The final state of the Sudoku grid.}
     \label{fig:sudoku_color_blank}
 \end{figure}
For $V_2$-machine programmed to fully color the graph in Fig. \ref{fig:color_sudoku}d, Fig. \ref{fig:sudoku_color_blank} shows the evolution of colors in the machine and Fig. \ref{fig:sudoku_color_blank}b shows the Sudoku grid at the end of the dynamics.
\begin{figure}
    \centering
    \includegraphics[width=1\linewidth]{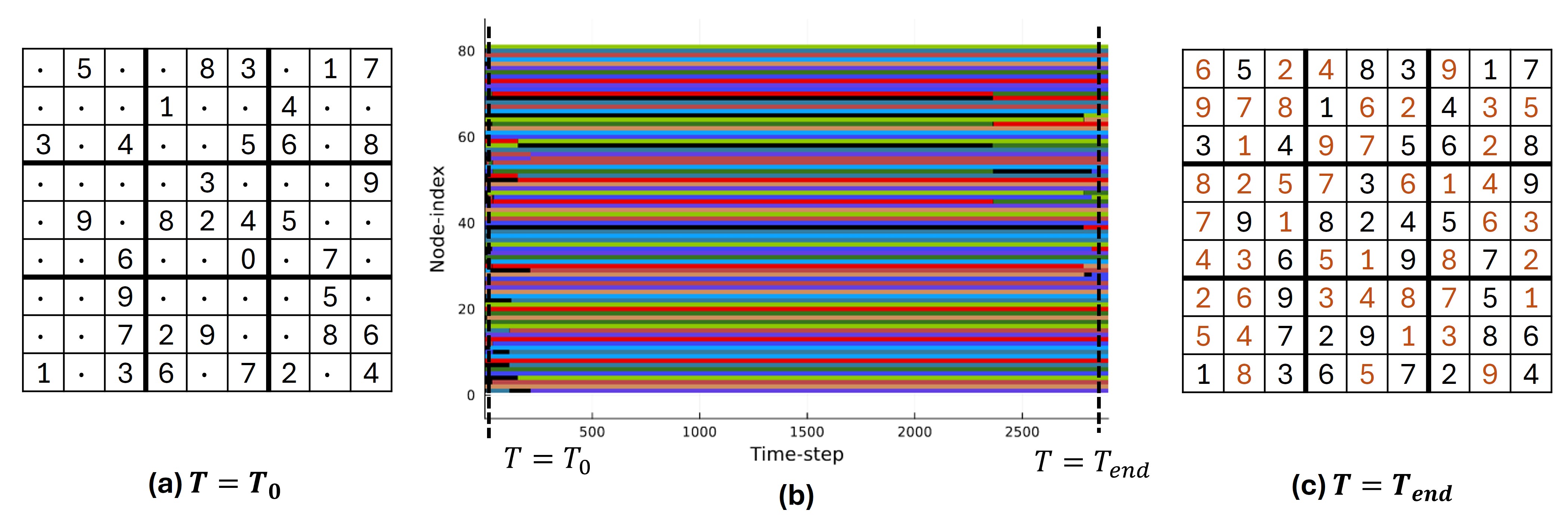}
    \caption{Solving Sudoku puzzle. (a) Sudoku puzzle. Coloring corresponding to the clues (at $T=T_0$) are held constant throughout the dynamics; (b) Color evolution; (c) Solution square as obtained from the coloring at $T=T_{end}$.}
    \label{fig:sudoku_color_puzzle}
\end{figure}
The same machine successfully solves the first 50 standard Sudoku Kaggle puzzles available at~\cite{Park2016}.
Figure~\ref{fig:sudoku_color_puzzle}a shows the puzzle we consider for illustration. 
Figure~\ref{fig:sudoku_color_puzzle}b shows the evolution of color, starting from a coloring corresponding to the puzzle in Fig.~\ref{fig:sudoku_color_puzzle}a, and
Fig.~\ref{fig:sudoku_color_puzzle}c shows the solution of the puzzle.

\section{Conclusion}

The growing interest in alternative models of computing has revived the Ising models as the basis for accelerated as well as scaled combinatorial optimization.
All NP-hard problems can be recast in terms of the problem of ground-state search of the Ising model and solved on realizations of the Ising model called Ising machines.

A new class of dynamical Ising machines, implemented on a variety of computing platforms, approach the search problem from the theoretical dynamical systems perspective by redefining the problem in terms of continuous spins and delivering the solution via its terminal state.
At the same time, the discrete nature of combinatorics puts tight constraints on the output (or the terminal state) of these dynamics. Since the Ising model is based on binary spins, variables of dynamical Ising machines must ultimately correspond to binary values.
Hence, the dynamical Ising machine, on the one hand, relies on the continuity of spins to be effective at navigating the energy landscape and, on the other, must terminate at states composed of binary spins.


The $V_2$ model sets itself apart from the current dynamical Ising machines
in that it does not put forth such requirements and is
allowed to converge to a highly non-binary state. The relation between the
non-binary states of the $V_2$ machine and, by requirement, the binary form
of a sought solution is straightforwardly established by a
hybrid binary-continuous representation (the relaxed spin). The key
property of the $V_2$ model is that the spin configuration contained in the
binary component of the relaxed spin representation in the machine's terminal state is of the same quality as an optimal rounding of the
underlying non-binary state.


As applications of the Ising machines to realistic problems are being
gradually expanded, concerted research effort is needed to develop Ising
models of such problems and establish these machines' efficacy at solving a diverse set of foundational problems. We demonstrate the computational
capabilities of the non-binary dynamical Ising machine based on the $V_2$
model through the coloring of non-planar complex graphs, namely those
corresponding to Latin squares and the Sudoku puzzle. Importantly, we show
that the $V_2$ model inherently converges to states satisfying the
feasibility constraint (definite color). This property signifies the
applicability of the Ising machine based on the $V_2$ model to a broad
class of high-level optimization problems where the feasibility requirement
is critical for the sole mapping of the problem to an Ising machine.



\bibliographystyle{naturemag}
\bibliography{references_AS,references_ME}


\end{document}